\documentclass[journal]{IEEEtran}
\usepackage{multirow}
\usepackage[flushleft]{threeparttable}
\usepackage{latexsym}
\usepackage{amsmath}
\usepackage{amssymb}
\usepackage{amsthm}
\usepackage{graphicx}
\usepackage{epsfig}
\usepackage{amssymb}
\usepackage{subcaption}
\usepackage{mwe}
\usepackage{float}
\usepackage{mathtools}
\usepackage{algorithm}
\usepackage[noend]{algpseudocode}
\usepackage[utf8]{inputenc}
\usepackage[english]{babel}
\usepackage{fancyhdr}
\usepackage{mleftright}
\usepackage{courier}
\usepackage{array}
\usepackage{booktabs}
\usepackage{syntonly}
\usepackage{xcolor}
\usepackage{colortbl}
\usepackage{lipsum}
\usepackage{epstopdf}
\usepackage{environ}
\usepackage{relsize}
\usepackage{multirow}
\usepackage{diagbox}
\usepackage{color} 
\usepackage[normalem]{ulem} 
\usepackage[makeroom]{cancel} 
\usepackage{bm}
\usepackage{enumitem}
\usepackage[normalem]{ulem}

\newtheorem{lemma}{Lemma}
\newtheorem{proposition}{Proposition}

\DeclareMathOperator*{\argmin}{argmin}

\makeatletter
\let\OldStatex\Statex
\renewcommand{\Statex}[1][3]{%
  \setlength\@tempdima{\algorithmicindent}%
  \OldStatex\hskip\dimexpr#1\@tempdima\relax}
\makeatother

\NewEnviron{myequation}[1][]{%
\begin{equation}
\scalebox{#1}{$\begin{aligned}\BODY\end{aligned}$}
\end{equation}
}
\NewEnviron{mymultiline}[1][]{%
\scalebox{#1}{\begin{align}
\BODY
\end{align}}
}
\newcommand{\vb}{\boldsymbol}
\newcommand{\mb}[1]{{\mathbf{#1}}}
\newcommand{\m}[1]{\mathrm{#1}}
\newcommand{\mt}[1]{\mathtt{#1}}
\newcommand{\ms}[1]{\mathsf{#1}}
\newcommand{\mc}[1]{{\mathcal{#1}}}

\newcommand{\PreserveBackslash}[1]{\let\temp=\\#1\let\\=\temp}
\newcolumntype{C}[1]{>{\PreserveBackslash\centering}p{#1}}
\newcolumntype{R}[1]{>{\PreserveBackslash\raggedleft}p{#1}}
\newcolumntype{L}[1]{>{\PreserveBackslash\raggedright}p{#1}}

\begin{document}
%

\title{Consensus Multi-Agent Reinforcement Learning for Volt-VAR Control in Power Distribution Networks}

\author{\IEEEauthorblockN{Yuanqi Gao, \textit{Student Member, IEEE}, Wei Wang, \textit{Student Member, IEEE}, and Nanpeng Yu, \textit{Senior Member, IEEE}}}

\markboth{Journal of \LaTeX\ Class Files,~Vol.~14, No.~8, August~2015}%
{Shell \MakeLowercase{\textit{et al.}}: Bare Demo of IEEEtran.cls for IEEE Journals}
%



\maketitle

\begin{abstract}




Volt-VAR control (VVC) is a critical application in active distribution network management system to reduce network losses and improve voltage profile. To remove dependency on inaccurate and incomplete network models and enhance resiliency against communication or controller failure, we propose consensus multi-agent deep reinforcement learning algorithm to solve the VVC problem. The VVC problem is formulated as a networked multi-agent Markov decision process, which is solved using the maximum entropy reinforcement learning framework and a novel communication-efficient consensus strategy. The proposed algorithm allows individual agents to learn a group control policy using local rewards. Numerical studies on IEEE distribution test feeders show that our proposed algorithm matches the performance of single-agent reinforcement learning benchmark. In addition, the proposed algorithm is shown to be communication efficient and resilient.


\end{abstract}


%
\IEEEpeerreviewmaketitle

\section{Introduction}
Volt-VAR control (VVC) determines the operation schedule of voltage regulating and VAR control devices to lower network losses, improve voltage profile, and reduce voltage violations \cite{Ahmadi2015VVO}. Traditional VVC adjust the tap positions of the on-load tap changers (OLTC) based on a line drop compensator (LDC), which models the voltage drop of the distribution line from the voltage regulator to the load center. However, the rapid growth of distributed energy resources makes it increasingly difficult to manage the voltage profile on active distribution networks.

To address the challenge of distribution system voltage control, a number of physical model-based and data-driven control methodologies have been proposed. The existing literature on VVC problem can be categorized into four groups according to the model assumption and the communication scheme: 1) model-based centralized, 2) model-based distributed, 3) data-driven centralized, and 4) data-driven distributed methods.

Model-based centralized methods assume that all distribution network measurements are collected by a central controller, which also has perfect knowledge of the distribution network parameters. The technical methods to solve the VVC problem include deterministic optimization, robust optimization, and meta-heuristic methods. The deterministic methods include dynamic programming \cite{Liang2001DP}, mixed-integer linear programming (MILP) \cite{borghetti2013using}, mixed-integer quadratically constrained programming (MIQCP) \cite{Ahmadi2015VVO}, and bi-level mixed-integer programming \cite{jha2019bi}. To account for the uncertainties in loads/DGs, robust VVC algorithms \cite{daratha2015robust} \cite{Zheng2017RobustRP} \cite{nazir2018two} have been developed. Meta-heuristic algorithms such as genetic algorithm \cite{Senjyu2008GA} and particle swarm optimization \cite{Yoshida2000PSO} have been adopted. 

To reduce the communication burden and enhance algorithms' resiliency against the failure of the centralized controller, model-based distributed algorithms for VVC have been studied. These methods include simulated annealing \cite{vaccaro2013voltage}, distributed decision making \cite{Farag2013novel}, and alternating direction method of multipliers (ADMM) considering the continuous relaxation of the discrete variables \cite{Robbins2016tap}. 

Model-based approaches assume complete and accurate physical network model, which are difficult to maintain for regional electric utilities. To overcome this problem, data-driven methods are deployed to determine control actions based on the operational data. A number of data-driven centralized methods have been proposed. In \cite{Bagheri2019Modelfree}, a $k$-nearest neighbor ($k$NN) regression model is used to estimate power loss and voltage change in response to the status change of VVC devices. Then, a heuristic approach is taken to determine the appropriate device status. In \cite{Pouriafari2019SVR}, a support vector regression (SVR) model is trained to approximate the power flow equation. The trained model is then embedded in a model predictive control (MPC) framework to obtain a one-day horizon VVC solution. Reinforcement learning (RL) and deep RL algorithms have also been developed for VVC. A batch RL algorithm that augments the historical dataset and trains a linear approximated action value function is proposed in \cite{Xu2019BatchRL}. The VVC problem is modeled as a constrained Markov decision process (CMDP) \cite{Wei2019Safe}. A safe off-policy RL algorithm is developed to avoid voltage violation while minimizing network losses and wear and tear of equipment.


Data-driven centralized methods are particularly advantageous when the distribution network model is unavailable. However, if the central controller fails, then the entire VVC system breaks down. Thus, extending data-driven centralized methods to enable decentralized communication and control will significantly improve the resiliency of the algorithm against individual controller or communication link failure. 

Very few data-driven decentralized VVC algorithms have been developed. Reference \cite{Xu2012Multiagent} proposes a multi-agent tabular Q-learning algorithm, in which the agents discover the global reward through a diffusion consensus protocol. Then the local Q values are updated by the standard Q-learning update. Reference \cite{zhang2020deep} developed a multi-agent deep Q-network (DQN) algorithm, which decouples the global action space into individual device's control space. However, the existing methods are either incapable of handling large state space or do not enable coordination between the individual agents.


In this paper, we propose a consensus multi-agent RL (C-MARL) algorithm for VVC in power distribution systems, which does not rely on accurate network model and handles state space with higher dimensionality. The proposed framework consists of a group of networked agents managing different VVC devices. Each agent learns two parametric models to approximate the global state value function and the local policy, respectively. These models are trained to maximize the agents' own expected cumulative local rewards, while minimizing the dissimilarity between their neighbors' and their own value functions in a communication-efficient manner. The performance of C-MARL is evaluated on three IEEE test feeders. The experimental results show that our proposed 
C-MARL algorithm is capable of learning a distributed Volt-VAR control policy that matches the performance of the single-agent RL benchmark. The proposed algorithm is resilient against the failure of individual agents and communications links. Furthermore, the proposed algorithm is much more communication-efficient than the ADMM-based consensus scheme.

The unique contributions of this paper are as follows:

$\bullet$ This paper proposes a novel consensus multi-agent RL algorithm to learn distributed Volt-VAR control policies from historical operational data.

$\bullet$ The C-MARL framework is communication efficient and significantly lowers the amount of data required to reach consensus.

$\bullet$ The proposed C-MARL algorithm matches the performance of the single-agent RL benchmark in terms of network and equipment operational costs.

$\bullet$ The proposed C-MARL algorithm is resilient against the failures of individual controllers and communication links.

The remainder of the paper is organized as follows: Section II presents the problem formulation of VVC. Section III provides the technical methods. Section IV discusses the setup and results of experimental studies. Section V states the conclusion.

\section{Problem Formulation}
In this section, we formulate the VVC problem as a networked multi-agent Markov decision process (MAMDP). We first introduce the concept of MAMDP, then we discuss the problem formulation.
\subsection{Basics of MAMDP}
A networked MAMDP \cite{FMARL2018} is a tuple $\mc{M} = (\mc{S}, \{\mc{A}^i\}_{i=1}^K, P, \{r^i\}_{i=1}^K, \mc{G}, \gamma)$ which consists of a global state space $\mc{S}$, $K$ local action spaces $\mc{A}^i$, a global state transition probability $P(s^\prime|s, a^1, a^2,\cdots,a^K)$ $\forall s,s^\prime \in \mc{S}, \forall a^i\in \mc{A}^i$, $K$ local reward functions $r^i(s, a^1, a^2, \cdots, a^K): \mc{S}\times \mc{A}^1\times\mc{A}^2\times \cdots \times \mc{A}^K\mapsto \mathbb{R}$, a communication network $\mc{G}=(\mc{V},\mc{E})$, and a discount factor $\gamma$. In a networked MAMDP, a set of $K$ learning agents select their local actions $A^{i}_t\in\mc{A}^i$ based on the current state $S_t\in\mc{S}$ at each discrete time step $t$. Then each of the agents receives a numerical reward $R^i_{t+1} = r^i(S_t,A^1_t,A^2_t,\cdots,A^K_t)$ and the environment's global state transitions to $S_{t+1}$ based on the state transition probability $P(S_{t+1}|S_t, A^1_t, A^2_t,\cdots,A^K_t)$. Also at time $t$, each agent $i$ can communicate and share its local information with its neighbors defined in the communication graph $\mc{G}$. In this paper, we assume $\mc{G}$ is connected. The neighbors of agent $i$ are denoted as $\mc{V}^i$. For notational simplicity, we denote the joint action and action space as $A_t = [A_t^1, A_t^2, \cdots, A_t^K]$ and $\mc{A} = \prod_{i=1}^K \mc{A}^i$, respectively. We also denote $R_{t+1} = r(S_t, A_t) = \frac{1}{K} \sum_{i=1}^K R_{t+1}^i = \frac{1}{K} \sum_{i=1}^K r^i(S_t,A_t)$ as the global averaged reward.

The goal of the networked agents is to find each agent's local control policy $\pi^i(a^i|s)$, such that the joint policy $\pi(a^1, a^2, \cdots, a^K|s)$ of all agents maximizes the expected discounted averaged return $J(\pi) = \mathbb{E}[G(\tau)]$, where $\tau$ is a trajectory of global states and global actions $S_0, A_0, S_1, A_1, \cdots$, and $G$ is the function that maps a trajectory to the discounted averaged return $G(\tau) = \sum_{t=0}^T \gamma^t \frac{1}{K} \sum_{i=1}^K R^i_{t+1}$. The local policy $\pi^i(a^i|s)$ represents a conditional probability distribution of local actions given the global state $s$. We assume the global policy is factored as $\pi(a^1, a^2, \cdots, a^K|s) = \prod_{i=1}^K \pi^i(a^i|s)$.

Two important functions for the multi-agent RL are the global state value function $v_\pi(s)$ and the global action value function $q_\pi(s,a)$ with respect to a given joint policy $\pi$. They are defined formally as:
\begin{align}
v_\pi(s) &= \underset{\tau \sim \pi}{\mathbb{E}}\left[\textstyle \sum_{k=0}^T\gamma^k R_{t+k+1}|S_t=s\right] \label{eq1}\\
\hspace*{-1cm}q_\pi(s,a) &= \underset{\tau \sim \pi}{\mathbb{E}}\left[\textstyle \sum_{k=0}^T\gamma^k R_{t+k+1}|S_t=s,A_t=a\right] \label{eq2}
\end{align}
$v_\pi(s)$ and $q_\pi(s,a)$ capture the expected return committing to a given policy for the starting state $s$ and action $a$. The optimal policy is thus the one that maximizes $v_\pi(s)$ for all $s$ (or maximizes $q_\pi(s,a)$ for all $s,a$).

In the next subsection, the distributed VVC problem will be formulated as a networked MAMDP.
\subsection{Formulate VVC as a MAMDP}
In this subsection, we first provide a brief introduction of the proposed multi-agent RL (MARL) VVC framework. Then we present the problem formulation.

We consider a radial distribution network whose node set is denoted as $\mc{N}$. The substation is numbered as $0$ and all other nodes are numbered as $1,\cdots,n$. The nodal voltage magnitude, real and reactive power at time $t$ of node $i\in\mc{N}$ is denoted as $V^i_t$, $p^i_t$, and $q^i_t$, respectively. Vectors $\mb{p}_t = [p^1_t, p^2_t, \cdots, p^n_t]$ and $\mb{q}_t = [q^1_t, q^2_t, \cdots, q^n_t]$ group all nodal real and reactive power injections except for the substation node.

Three types of VVC devices are considered in this paper:

$\bullet$ A voltage regulator is placed at the substation (reference node). Thus, the reference voltage of the network at time step $t$ can take on several discrete values $V^{0}_t = 1 \m{p.u.} + x^{\m{reg}}_t \cdot M^{\m{reg}}$ according to the tap position $x_t^{\m{reg}}$ and the fixed step size $M^{\m{reg}}$. The numerical values will be provided in Section IV.

$\bullet$ A capacitor bank's reactive power output $q^{i,\m{cap}}_t$ is determined by its on/off status and nodal voltage as: $q^{i,\m{cap}}_t = x_t^{\m{cap}} \cdot M^{\m{cap}} \cdot (V^i)^2$. $x_t^{\m{cap}}\in\{0,1\}$ denotes the on/off status. $M^{\m{cap}}$ denotes the rated reactive power of the capacitor.


$\bullet$ An on-load tap changer (OLTC) is modeled as an ideal transformer with a variable turns ratio. When an OLTC is present on a branch $(i,j)$, its branch power flow is described by (\ref{oltc}) in the DistFlow equation:
\begin{align}\label{oltc}
(V^{j}_t)^2/a_t^2 = (V^{i}_t)^2 - 2r^{ij}p^{ij}_t - 2x^{ij}q^{ij}_t + [(r^{ij})^2 + (x^{ij})^2]l^{ij}_t
\end{align}
where $p^{ij}_t$ and $q^{ij}_t$ are the branch power flow and $l^{ij}_t$ is the square of branch current. The turns ratio at time $t$ is given by $a_t = 1+x_t^{\m{tsf}}\cdot M^{\m{tsf}}$, where $x_t^{\m{tsf}}$ is the tap position and $M^{\m{tsf}}$ denotes the step size. 

\begin{figure}[h]
	\centering
	\includegraphics[width=8.5cm]{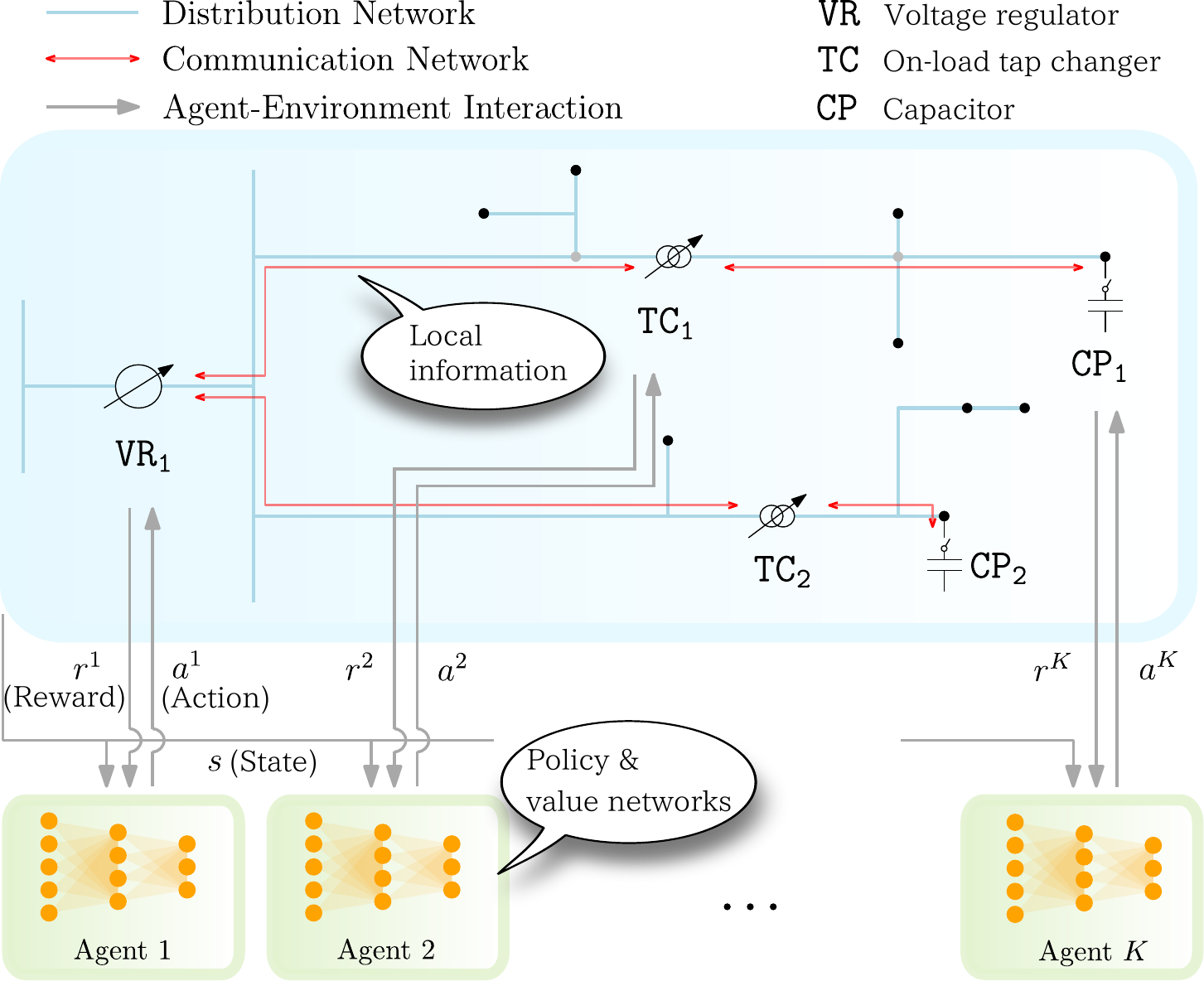}
	\caption{The proposed MARL based VVC framework}
	\label{fig_overview}
\end{figure}
An overview of the proposed MARL based VVC framework is shown in Fig. \ref{fig_overview}. Each agent is associated with one VVC device and determines its own control actions. Therefore, for the substation voltage regulator agent $1$, the local action is the discrete tap number $A^1_t = x^{\m{reg}}_{t+1}$. The subscript $t+1$ in $x^{\m{reg}}_{t+1}$ designates that it is the new tap position after action $A^1_t$ is taken. The local action spaces for capacitor agents and OLTC agents are defined in a similar manner. The design of the local reward functions should satisfy two requirements: 1) the averaged rewards $\frac{1}{K}\sum_{i=1}^K r^i(s,a)$ should reflect the networked agents' VVC objective and 2) each local reward must be calculated based on the local metering data received by the corresponding agent. To this end, we define the local reward function as
\begin{align}\label{local_reward}
R^i_{t+1} = r^i(S_t, A_t) = -C^l \sum_{\ell\in \mc{L}_i} p^{l,\ell}_t - C^s|x^i_{t}-x^i_{t+1}| - \bar{\lambda} C^i_{t+1}
\end{align}
where $p^{l,\ell}_t$ is the real power loss on branch $\ell$ after the joint action $A_t$ is taken; $\mc{L}_i$ is the set of branches metered by agent $i$; $C^l$ and $C^s$ are the costs associated with power loss and devices' switching actions, respectively. $x^i_t$ is the generic term of the discrete control. For example, $x^i_t = x^{\m{reg}}_t$ if agent $i$ is associated with a voltage regulator. The term $C^i_{t+1}$ describes the voltage constraint violation and $\bar{\lambda}$ is the associated penalty factor. The constraint violation is given by:
\begin{align}\label{local_cv}
C^i_{t+1} = c^i(S_t, A_t) = \sum_{k\in \mc{N}_i}[\mathbb{I}(V^k_t>\bar{V}) +\mathbb{I}(V^k_t<\underline{V})]
\end{align}
where $\mathbb{I}$ is the indicator function; $\mc{N}_i$ is the set of nodes metered by agent $i$; $V^k_t$ represents the voltage magnitude followed by the joint action $A_t$. The exact formulation of the sets $\mc{L}_i$ and $\mc{N}_i$, as well as various parameters $C^l$, $C^s$, and $\lambda$ will be described for each test feeder in Section IV. The global state at time $t$ is defined as $S_t = [\mb{p}_t, \mb{q}_t, A_{t-1}, t]$. That is, the global state contains the network power injections, the existing VVC devices' status at the previous time step, and a discrete time step $t$. Finally, we choose a global discount factor $\gamma$ that is less than one. This completes the formulation of the distributed VVC problem as an MAMDP.

With this MAMDP formulation, we can interpret the value functions (\ref{eq1}) and (\ref{eq2}) in terms of VVC as follows: At each time step $t$, the networked agent's goal is to minimize the long term discounted operational cost and the constraint violation. This long term objective does not easily break into a set of unrelated single time step objectives, because the cost of device switching links the goal of adjacent time steps.

In the next section, we present the technical details of the multi-agent RL algorithm.

\section{Technical Methods}
In this section, we present the proposed consensus multi-agent deep RL-based VVC algorithm. We derive the proposed algorithm in three stages. First, we review the preliminary of centralized off-policy maximum entropy RL framework. Then we reformulate this framework into a distributed multi-agent framework. Finally, we present the proposed communication-efficient 
C-MARL algorithm to solve the VVC problem.

\subsection{Off-policy Maximum Entropy RL}
In maximum entropy RL, the policy maximizes both the return and the entropy of the policy \cite{SAC}. In the context of data-driven VVC, the policy entropy maximization is introduced for two reasons. First, without an accurate physical model, all data-driven methods must involve some sort of exploration \cite{Bagheri2019Modelfree}. That means, it must try different control actions before becoming informed about which of them is the best. To this end, maximum entropy RL provides an efficient and principled way for balancing the exploration and exploitation \cite{sutton2018reinforcement}. Second, an off-policy algorithm can be derived within the maximum entropy RL framework. Off-policy RL algorithms are capable of learning from past experiences so that it can be trained using a much smaller amount of samples collected from the distribution grid. Next, we provide a mathematical characterization of the optimal policy in maximum entropy RL. This is critical to the development of off-policy RL algorithms.

The maximum entropy RL regularizes the reward function by the entropy of the policy $r(s,a) + \alpha H(\pi(\cdot|s))$. $\alpha$ is a temperature parameter that determines the contribution of the entropy to the reward. The state value function in this case is defined as:
\begin{align}\label{eq_def_v}
\ms{v}_\pi(s) = \underset{\tau \sim \pi}{\mathbb{E}}\left[\textstyle \sum_{k=0}^T\gamma^k (R_{t+k+1} + \alpha H(\pi(\cdot|S_{t+k}))) |S_t=s\right]
\end{align}
where $\ms{v}_\pi(s)$ denotes the entropy regularized value function. It follows from (\ref{eq_def_v}) that a policy, which maximizes $\ms{v}_\pi(s)$ is maximizing the combined return and policy entropy. The latter maintains a certain level of stochasticity of the policy. Thus we can balance the exploration and exploitation by following the current policy $\pi$ throughout the learning process.

The Bellman equation for $\ms{v}_\pi(s)$ is derived as:
\begin{align}
\ms{v}_\pi(s) & = \mathbb{E}_{a\sim\pi}\mathbb{E}_{s^\prime\sim P}\left[ r+\gamma \ms{v}_\pi(s^\prime)\right] + \alpha H(\pi(\cdot|s)) \nonumber \\
& = \mathbb{E}_{a\sim\pi}\left[r + \gamma\mathbb{E}_{s^\prime\sim P}\left[ \ms{v}_\pi(s^\prime)\right] - \alpha \log\pi(a|s)\right]\label{eq_bellman}
\end{align}
The definition of entropy is used to derive the second equality. The optimal entropy-regularized state value function is defined as $\ms{v}_*(s) = \max_\pi \ms{v}_\pi(s)$. Similarly, the optimal entropy-regularized policy $\pi^*$ is defined as the one whose entropy-regularized value function is $\ms{v}_*(s)$. We term the two-tuple $(\ms{v}_*, \pi^*)$ an optimality pair, which is shown to be the solution to the off-policy consistency equation \cite[Corollary 21]{BGap}:
\begin{align}\label{eq_bellman_off}
\ms{v}(s) = r(s,a) + \gamma \mathbb{E}_{s^\prime\sim P}\left[\ms{v}(s^\prime)\right] - \alpha \log\pi(a|s)\quad \forall s,a
\end{align}
(\ref{eq_bellman_off}) characterizes the optimal policy and motivates our subsequent algorithm developments. However, it is very challenging to solve (\ref{eq_bellman_off}) directly due to high-dimensional and continuous state space. In addition, (\ref{eq_bellman_off}) is stated in a centralized format. In the next subsection, we derive a distributed and off-policy algorithm to approximate the solution to (\ref{eq_bellman_off}) in a sample-efficient manner. To make the learning tractable, we restrict the class of functions we consider for the value function and the policy function.

\subsection{Distributed Optimization}
In this subsection, we transform the problem of finding optimal local policy and state value function as a distributed consensus optimization problem of the following form:
\begin{align}\label{general_opt}
 \underset{\vb{w}}{\min} \quad \textstyle \sum_{i=1}^K J^i(w^i) \quad \text{s.t.}\; w^1 = w^2 = \cdots = w^K
\end{align}
where $J^i$ are the local objective functions. (\ref{general_opt}) appears ubiquitously in distributed adaptive learning \cite{diffusion}, distributed algorithms for linear algebraic systems \cite{Peng2018communication}, and distributed parameter estimation \cite{Xie2012fully}.

We first approximate the solution of optimal policy and value function with the following stochastic nonlinear program, which is commonly done in deep RL literature \cite{mnih2015human}:
\begin{align}\label{cen_opt}
\min_{\ms{v},\pi} \;  \underset{s,a\sim \mc{D}}{\mathbb{E}}\left(\ms{v}(s) - \{r+\gamma\mathbb{E}_{s^\prime\sim P}\left[ \ms{v}(s^\prime)\right] - \alpha \log\pi(a|s) \} \right)^2
\end{align}
$\mc{D}$ is the data distribution, which will be approximated by an experience replay buffer \cite{mnih2015human}. Next, we define $w^i$ in (\ref{cen_opt}) as each agent's local copy of the global value and policy functions $\ms{v}^{i}(s)$ and $\pi^i(a|s)$. At optimality, these local functions need to reach consensus. Thus $\ms{v}^{i} = \ms{v}^{j}$ and $\pi^i = \pi^j, \forall i,j$ constitute the constraints in (\ref{general_opt}). Although the joint policy $\pi^i(a|s)$ is maintained by all agents, only the $i$-th coordinate $\pi^i(a^i|s)$ (the $i$-th local action) is actuated by agent $i$. Note that $\ms{v}^i$ and $\pi^i$ are infinite dimensional.

Now to decompose the objective function in (\ref{cen_opt}), we first declare two sets of functions for later derivations. Let $f(s,a) = \ms{v}(s) - \gamma\mathbb{E}_{s^\prime\sim P}\left[ \ms{v}(s^\prime)\right] + \alpha \log\pi(a|s)$ and $f^i(s,a) = \ms{v}^i(s) - \gamma\mathbb{E}_{s^\prime\sim P}\left[ \ms{v}^i(s^\prime)\right] + \alpha \log\pi^i(a|s)$. $\zeta(s,a) = [\ms{v}(s), \pi(a|s)]^T$ and $\zeta^i(s,a)=[\ms{v}^i(s), \pi^i(a|s)]^T$.
Then, the minimization problem (\ref{cen_opt}) can be written as:
\begin{align}
  & \argmin_{\zeta} \;  \underset{s,a,r\sim \mc{D}}{\mathbb{E}}\big(f(s,a) - r \big)^2 \\
= & \argmin_{\zeta} \;  \underset{s,a,r\sim \mc{D}}{\mathbb{E}}f(s,a)^2 - 2rf(s,a) + r^2 \\
= & \argmin_{\bm{\zeta} \in \Omega} \;  \underset{s,a,r^i\sim \mc{D}}{\mathbb{E}}\frac{1}{K}\sum_{i=1}^K f^i(s,a)^2 - \nonumber \\
 & \hphantom{some invisible char} \frac{1}{K}\sum_{i=1}^K 2r^if^i(s,a) + \frac{1}{K}\sum_{i=1}^K {(r^{i})}^2 \\
= & \argmin_{\bm{\zeta} \in \Omega} \;  \underset{s,a,r^i\sim \mc{D}}{\mathbb{E}} \frac{1}{K}\sum_{i=1}^K \big( f^i(s,a) - r^i \big)^2 \label{decen_opt}
\end{align}
where $\bm{\zeta} = [(\zeta^1)^T, (\zeta^2)^T, \cdots, (\zeta^K)^T]^T$. $\Omega$ is the set containing all $\bm{\zeta}$ such that $\zeta^1 = \zeta^2 = \cdots = \zeta^K$. 
Using the degree matrix $D$ and the adjacency matrix $A$ of $\mc{G}$, the constraints can be rewritten as $(D_{ii}\otimes I_2) \zeta^i = (A_{i}\otimes I_2) \bm{\zeta}$, $\forall i$. $I_2$ is the identity matrix of size 2, $\otimes$ designates Kronecker product, and $A_i$ denotes the $i$th row of $A$. We will use the notations $\bar{D}=D\otimes I_2$ and $\bar{A}=A\otimes I_2$, with $\bar{D}_{ii}$ and $\bar{A}_i$ being understood as the $ii$-th block and $i$-th block row, respectively.
(\ref{decen_opt}) decomposes the global learning objective and is compatible with the MAMDP model. Specifically, each agent $i$ receives the local reward $r^i$ and takes local actions $a^i$. The consensus is achieved through neighbor-to-neighbor communication. 

To derive a tractable learning algorithm for optimization problem (\ref{decen_opt}), we parameterize $\ms{v}^i$ and $\pi^i$ as function approximators such as deep neural networks (NN): $\ms{v}_{\psi_i}\approx \ms{v}^i$ and $\pi_{\phi_i}\approx \pi^i$ with the parameters of deep NNs denoted by $\varphi_i = [\psi_i,\phi_i]$. 
Additionally, 
we use a separate target network $\ms{v}_{\bar{\psi}_i}$ for evaluating $\ms{v}^i(s^\prime)$ \cite{SAC}. We denote the $\varphi_i, \bar{\psi}_i$-parameterization of $f^i$ and $\zeta^i$ as $\bar{f}_{\varphi_i}$ and $\zeta_{\varphi_i}$. 
 
Therefore, (\ref{decen_opt}) can be rewritten as:
\begin{equation}\label{mamdp_problem}
\begin{aligned}
\underset{\varphi_i}{\min} &
& & \underset{s,a,r^i,s^\prime \sim \mc{D}}{\mathbb{E}} \frac{1}{K}\sum_{i=1}^K \big(\bar{f}_{\varphi_i}(s,a,s^\prime) - r^i \big)^2 \\
\text{s.t.} &
& &  \bar{D}_{ii}\zeta_{\varphi_i}(s,a) = \bar{A}_i \bm{\zeta}_{\bm{\varphi}}(s,a) \quad \forall i\in\mc{V}, s, a
\end{aligned}
\end{equation}
where $\bm{\zeta}_{\bm{\varphi}} = [\zeta_{\varphi_1}^T,\zeta_{\varphi_2}^T,\cdots, \zeta_{\varphi_K}^T]^T$. Note that, (\ref{mamdp_problem}) has a finite number of decision variables and infinitely many constraints. This infinite constraint set can be reformulated as a finite one, $\varphi_i=\varphi_j$, $\forall i, j$. Several methodologies such as ADMM \cite{ADMM01} and diffusion adaptation strategies \cite{diffusion} can be used to solve (\ref{mamdp_problem}) with the finite constraint set. However, it is extremely costly to communicate the full set of parameters of deep neural networks. To address this problem, in the next subsection, we leave the constraints as they are and derive a stochastic approximation type algorithm to solve (\ref{mamdp_problem}), which significantly improves the communication efficiency.


\subsection{Communication-Efficient Multi-Agent Policy Consensus}
The goal of this subsection is to approximate the solution to (\ref{mamdp_problem}) by randomization. That is, we randomly enforce a subset of all constraints in each iteration. First, we adopt the model for semi-infinite programming in \cite{tadic2006randomized} to convert the infinite constraint set into a finite one. Specifically, we approximate the constraints represented in (\ref{mamdp_problem}) as the following stochastic programming representation:
\begin{align}\label{stochastic_representation}
    \int_{s,a\in\mc{S}\times\mc{A}} \mb{h}\big(\bar{D}_{ii}\zeta_{\varphi_i}(s,a) - \bar{A}_i\bm{\zeta}_{\bm{\varphi}}(s,a)\big) d \mu_{\mc{SA}} = \mb{0}
\end{align}
for all $i\in\mc{V}$. $\mb{h}(x,y)=[h(x),h(y)]$ stacks two penalty functions $h$, which satisfies $h(0)=0$ and $h(x)>0, \forall x\neq 0$. $\mu_{\mc{SA}}$ is a probability measure defined on the global state-action space $\mc{S}\times \mc{A}$. Assuming the continuity of the function $\zeta_{\varphi_i}$ and $h$, as well as full support assumption of $\mu_{SA}$, we can establish the following propositions:
\begin{proposition}\label{prop1}
Let $d$ be a metric on $\mc{S}\times \mc{A}$. Assume $\zeta_{\varphi_i}$ is continuous with respect to $d$ for every $\varphi_i$ and $h$ is continuous with respect to the Euclidean metric. Also assume $\mu_{\mc{SA}}(X)>0$ for every non-empty open subset $X\subseteq \mc{S}\times \mc{A}$. Then the constraint set in (\ref{mamdp_problem}) is equivalent to (\ref{stochastic_representation}).
\end{proposition}

\begin{proposition}\label{prop2}
With the same assumptions in Proposition \ref{prop1} except for the continuity assumption about $\zeta_{\varphi_i}$. Then every feasible point of (\ref{stochastic_representation}) satisfies most of the constraints in (\ref{mamdp_problem}), except for a subset of measure zero.
\end{proposition}
Proposition \ref{prop1} and \ref{prop2} are theoretically reassuring. In practice, $\mu_{\mc{SA}}$ will be approximated by $\mc{D}$, which is the data distribution in (\ref{mamdp_problem}).

Consider the quadratic penalty for non-consensus $h(x)=\frac{1}{2}x^2$. Under this approximated stochastic programming representation, the Lagrangian of (\ref{mamdp_problem}) is given by:
\begin{align}\label{lagrange}
\mc{L}(\bm{\varphi},\bm{\lambda}) = \frac{1}{K}\sum_{i=1}^K 
\mc{L}^i(\bm{\varphi},\lambda_i) \hphantom{----------}
\end{align}
\begin{align}\label{lagrange_i}
\mc{L}^i(\bm{\varphi},\lambda_i) = 
\underset{s,a,r^i,s^\prime \sim \mc{D}}{\mathbb{E}} &\Big( \big(\bar{f}_{\varphi_i}(s,a,s^\prime) - r^i \big)^2 \nonumber \\
+\frac{\lambda_i}{2}&||\bar{D}_{ii}\zeta_{\varphi_i}(s,a) - \bar{A}_i\bm{\zeta}_{\bm{\varphi}}(s,a)||^2\Big)
\end{align}
The primal variables $\bm{\varphi}$ and the multipliers $\bm{\lambda}$ can be solved by the primal-dual method. However, we found that using a fixed $\bm{\lambda}$ parameter achieves good empirical performance. The detailed value for the multipliers will be provided in Section IV. (\ref{lagrange_i}) has a tractable sample gradient and can be readily tackled by established deep learning routines such as stochastic gradient descent (SGD). Specifically, each agent performs the minimization of sample-estimated $\mc{L}^i(\bm{\varphi},\lambda_i)$. In addition, similar to the adapt-then-combine (ATC) algorithm \cite{sayed2014adaptation}, we first perform the minimization of the first term in (\ref{lagrange_i}), then use the immediately updated weights to evaluate and minimize the second term:
\begin{align}
    &\tilde{\varphi}^\nu_i = \varphi_i^\nu - \eta\nabla_{\varphi_i} \big(\bar{f}_{\varphi^\nu_i}(s,a,s^\prime) - r^i \big)^2 \label{1st_step}\\
    &\varphi^{\nu+1}_i = \tilde{\varphi}_i^\nu - \eta\frac{\lambda_i}{2}\nabla_{\varphi_i} ||\bar{D}_{ii}\zeta_{\tilde{\varphi}^\nu_i}(s,a) - \bar{A}_i\bm{\zeta}_{\bm{\varphi}^\nu}(s,a)||^2 \label{2nd_step} \\
    & \bar{\psi}^{\nu+1}_i = \rho \bar{\psi}^{\nu}_i + (1-\rho)\psi^{\nu+1}_i \label{3rd_step}
\end{align}
where $\nu$ is the iteration count and $\rho$ is an exponential smoothing parameter. $s,a,r^i,s^\prime$ are sampled data from the experience replay $\mc{D}$, which is assumed to be initialized by the historical data. When conducting the update in (\ref{2nd_step}), a communication of each agent $i$ with its neighbors is established. The information being transmitted includes $s,a$ and $\zeta_{\varphi_j}(s,a)$. The full algorithm is summarized in Algorithm \ref{algo1}.
\begin{algorithm}[h]
	\caption{C-MARL for VVC}
	\label{algo1}
	\textbf{Input:} Historical dataset $\mc{D}$, update frequency $C$, communication graph $\mc{G}$
	\begin{algorithmic}[1]
	\For{$i=1,\cdots,K$}{}
		\State Initialize $\varphi^0_i=[\psi^0_i,\phi^0_i], \bar{\psi}_i^0$
	\EndFor
	\For{$\nu=0,\cdots,$}{}
		\State Sample $i$ from $[1,2,\cdots,K]$ uniformly
		\State Sample mini-batch $\mc{B}=\{(s,a,r^i,s^\prime)\}$ from $\mc{D}$
		\State Update $\varphi^\nu_i$ by (\ref{1st_step})
		\State Collects $\zeta_{\varphi_j}(s,a)$ from $i$'s neighbors.
		\State Update $\varphi^\nu_i$ by (\ref{2nd_step})
		\State Update $\bar{\psi}_i^\nu$ by (\ref{3rd_step})
		\If{$\m{mod}(\nu,K\cdot C)=0$}{}
		\For{$i=1,\cdots,K$}{}
		    \State Take control actions $A^i_t\sim \pi_{\phi^{\nu+1}_i}(\cdot|S_t)$
		    \State $\mc{D} = \mc{D} \cup \{(S_t, A^i_t,R^i_{t+1},S_{t+1})\}$
		\EndFor
		\EndIf
	\EndFor
	\end{algorithmic}
\end{algorithm}
The proposed 
C-MARL algorithm proceeds as follows: First all agents initialize their deep NN parameters. Then the agents communicate and update their local policy and value functions according to the scheme described in (\ref{1st_step})-(\ref{3rd_step}). We let each agent communicate and update $C$ times (on average) between adjacent control actuation steps $t$ and $t+1$. At time $t$, all agents take their control actions (tap positions of the voltage regulating devices), and store the transition information into the experience replay buffer $\mc{D}$.

\subsection{Algorithm Implementation}
This subsection provides additional implementation details for the proposed C-MARL VVC algorithm. We will discuss the NN architecture design and variable encoding in these NNs.

$\bullet$ $\ms{v}_{\psi_i}(s)$ and $\ms{v}_{\bar{\psi}_i}(s)$ networks: the value networks are the standard multilayer perceptron. The input of the networks are the global state $s$ and the output is the value of that state.

$\bullet$ $\pi_{\phi_i}(a|s)$ networks: we adopt the device-decoupled network structure \cite{Wei2019Safe}, which divides the outputs of the policy network into $K$ groups. The output neurons in each group corresponds to the local action space $|\mc{A}^i|$ for each device. In addition, we adopt the ordinal encoding layer \cite{DCAOP} for each group to represent the order information of the devices' tap positions. The hidden layers are shared by all groups.

$\bullet$ Encoding the global time step $t$: in this study, we only encode the hour-of-week part of the global time step $t$, which ranges from 0 to 167. $t$ is encoded in two coordinates $[\cos(2\pi t/168), \sin(2\pi t/168)]$ to reflect its periodic nature.


\section{Numerical Studies}
The numerical studies of the proposed C-MARL algorithm are conducted on three test feeders. The experimental setup for the three test feeders are provided in Section IV.A. The sample efficiency, communication efficiency, and resiliency of the proposed algorithm are validated in Section IV.B.

\subsection{Numerical Setup}
\subsubsection{Distribution Networks and Nodal Power Data} The IEEE 4-bus, 34-bus, and 123-bus distribution test feeders \cite{testfeeder} are used in the numerical studies. The VVC devices are setup on these test feeders as follows: For all test feeders, a voltage regulator ($\mt{VR_1}$) is located at the substation node and controls the reference voltage. Voltage regulators have $21$ tap positions with step size $M^{\m{reg}}=0.005$, which evenly divides the turns ratios between $0.95$ and $1.05$. We assume the same tap position configuration for the OLTCs. For the 4-bus feeder, an OLTC is placed between node 2 and 3 ($\mt{TC_1}$) and a capacitor with rating $M^{\m{cap}}=200$ kVar is placed at node 4 ($\mt{CP_1}$). For the 34-bus test feeder, two OLTCs are placed between node 814 and 850 ($\mt{TC_1}$), and node 852 and 832 ($\mt{TC_2}$). Two capacitors are placed at node 844 ($\mt{CP_1}$: 100 kVar) and node 847 ($\mt{CP_2}$: 150 kVar). For the 123-bus, three OLTCs are placed between node 10 and 15 ($\mt{TC_1}$), node 67 and 160 ($\mt{TC_2}$), and node 25 and 26 ($\mt{TC_3}$). Four capacitors are placed at node 83 ($\mt{CP_1}$: 200 kVar), node 88 ($\mt{CP_2}$: 50 kVar), node 90 ($\mt{CP_3}$: 50 kVar), and node 92 ($\mt{CP_4}$: 50 kVar). The initial turns ratios of voltage regulators and OLTCs are 1. Initially, the capacitors are switched off. 

The time series of hourly load data are obtained from the London smart meter dataset \cite{loaddata}. The dataset contains one year of half-hourly smart meter kWh measurements from approximately 5,000 customers. The measurements are aggregated and scaled to match the test feeders' loading level. The final load data have the same spatial load distribution and power factors as that of the IEEE standard test cases.

\subsubsection{Local Reward Setups and Communication Networks}
The parameters that appear in the local reward (\ref{local_reward}) and local operation constraint violation (\ref{local_cv}) are as follows: For all test cases, the cost of electricity, cost per switching action, and the constraint violation penalty are set as $C^l = \$0.04/\m{kWh}$, $C^s = \$0.1$, and $\bar{\lambda}=2C^l$, respectively. The voltage bounds are $\bar{V}=1.05$ and $\underline{V}=0.95$ p.u. The capacitors meter the voltage at its own node and the line real power loss within one-degree neighbors. The voltage regulators meter the voltage at the first downstream node from the substation. The OLTCs meter the voltage at its first downstream node and the power loss on its branch. A fixed (time-invariant) communication graph is assumed for each of the test feeders. The neighbor-to-neighbor relationships are summarized in Table \ref{table_communication}. Each line represents a bi-directional communication link.
\begin{table}[h]
\setlength{\tabcolsep}{1pt}
    \centering
	\caption{Communication Networks}
	\begin{tabular}{L{2.78cm}L{2.78cm}L{2.78cm}}
		\toprule  
		 4-bus & 34-bus & 123-bus \\ \hline \vspace{0.1cm}
	     \includegraphics[width=0.44\textwidth]{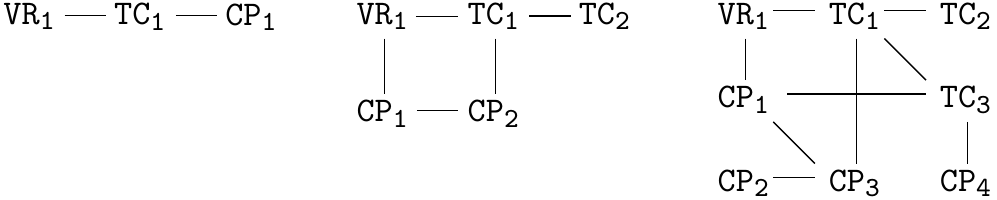} &  & \\
		\bottomrule
	\end{tabular}
	\label{table_communication}
\end{table}

\subsection{Algorithm Setup}
In the numerical studies, we compare the performance of our proposed algorithm with two benchmarks: the single-agent SAC \cite{SAC} and the multi-agent off-policy RL using the linearized ADMM consensus strategy \cite[Algorithm 1]{DLM2015}. The single-agent SAC serves as a stability baseline and the ADMM is used for comparison purpose.

$\bullet$ For the single-agent SAC, the reward is defined as the average of the local rewards. The agent's action is defined as the union of the local actions.

$\bullet$ For the linearized ADMM, the optimization variables for each agent are the deep NN parameters $\varphi_i$. We maintain a separate deep NN $\zeta_{\underline{\varphi}_i}$, whose structure is the same as $\zeta_{\varphi_i}$ and the parameters $\underline{\varphi}_i$ are the local dual variables. The same target network construct for evaluating $\ms{v}(s^\prime)$ is adopted.

The hyperparameters of the algorithms are provided in Table \ref{table_params}. The hyperparameters of the algorithms are tuned individually to reach their best performance. The last row of Table \ref{table_params} shows the parameters shared by all algorithms. Note that we scale the reward function to match the weights of neural networks. If not specified otherwise, these parameters will be used for all the numerical studies. Three parameters in the curly brackets are for the three distribution networks, from left to right, 4-bus, 34-bus, and 123-bus, respectively.
\begin{table}[h]
	\caption{Hyperparameters of Benchmark and Proposed Algorithms}
	\begin{tabular}{L{1.5cm} L{3.5cm} L{2.5cm}}
		\toprule  
		\multirow{2}{*}{SAC} & temperature parameter $\alpha$ & $\{0.5,0.2,0.1\}$  \\ 
		    & learning rate & $0.001$  \\ 
		    & number of hidden units & $\{64, 80, 128\}$ \\
		    & smoothing parameter $\rho$ & $0.99$ \\ 
		    & minibatch size & $16$ \\\hline    
		\multirow{2}{*}{ADMM} & temperature parameter $\alpha$ & $\{0.5,0.2,0.1\}$  \\ 
		    & $c$ in \cite{DLM2015} & $1$ \\ 
		    & $\rho$ in \cite{DLM2015} & $500$ \\ 
		    & number of hidden units & $\{32, 64, 64\}$ \\
		    & smoothing parameter $\rho$ & $0.99$ \\ 
		    & minibatch size & $16$ \\\hline        
		\multirow{2}{*}{C-MARL} & temperature parameter $\alpha$ & $\{0.5,0.2,0.1\}$  \\ 
		    & learning rate & $0.001$  \\ 
		    & number of hidden units & $\{32,64,128\}$ \\
		    & smoothing parameter $\rho$ & $0.99$ \\ 
		    & minibatch size & $16$ \\\hline        
		\multirow{2}{*}{shared} &discount factor & $0.95$  \\
		       & update frequency $C$ & 1 \\
		       & consensus parameter $\lambda_i$ & 1 \\
		       &number of hidden layers & 2 \\
		       &hidden unit nonlinearity & tanh  \\
		       &optimizer & Adam \\
		       &reward scale & 5 \\
		\bottomrule
	\end{tabular}
	\label{table_params}
\end{table}

\subsection{Stability, Sample Efficiency, and Communication Efficiency}
In this subsection, we report the stability, sample efficiency, and communication efficiency of the proposed and benchmark VVC algorithms. The average of the hourly rewards in (\ref{local_reward}) and the average of the constraint violations in (\ref{local_cv}) versus the number of training samples and the number of transmitted data points are shown in Fig. \ref{fig_bus4}-\ref{fig_bus123}. The horizontal axis beneath the plots shows the number of training samples of the form $(S_t,A_t,R_{t+1},S_{t+1})$. For the proposed C-MARL algorithm, the data being transmitted include the global time steps $\{t\}$ and the corresponding values $\ms{v}_{\psi_i}(\{S_t\}), \pi_{\phi_i}(\{A_t\}|\{S_t\})$ of the mini-batch; for the ADMM consensus strategy, the data being transmitted are the neural network weights $\varphi_i = [\psi_i, \phi_i]$. For all figures, the solid curve represents the median of five independent runs; the shaded areas are the upper and lower error bounds.

As shown in Fig. \ref{fig_bus4}-\ref{fig_bus123}, all three algorithms' performances stabilize after a certain amount of training samples are collected and used for training. Our proposed algorithm achieves a similar level of performance as the single-agent benchmark in all test cases in terms of hourly reward and constraint violation. This demonstrates the effectiveness of the proposed randomized consensus protocol. The proposed algorithm yields significant improvement on communication efficiency compared with the ADMM consensus protocol. In addition, the communication cost of our proposed algorithm stays constant across the test feeders. This is because only the sample data are transmitted. The communication burden in the ADMM consensus strategy grows quickly with the size of the physical network and the number of agents.
\begin{figure}[h]
	\centering
	\includegraphics[width=9cm]{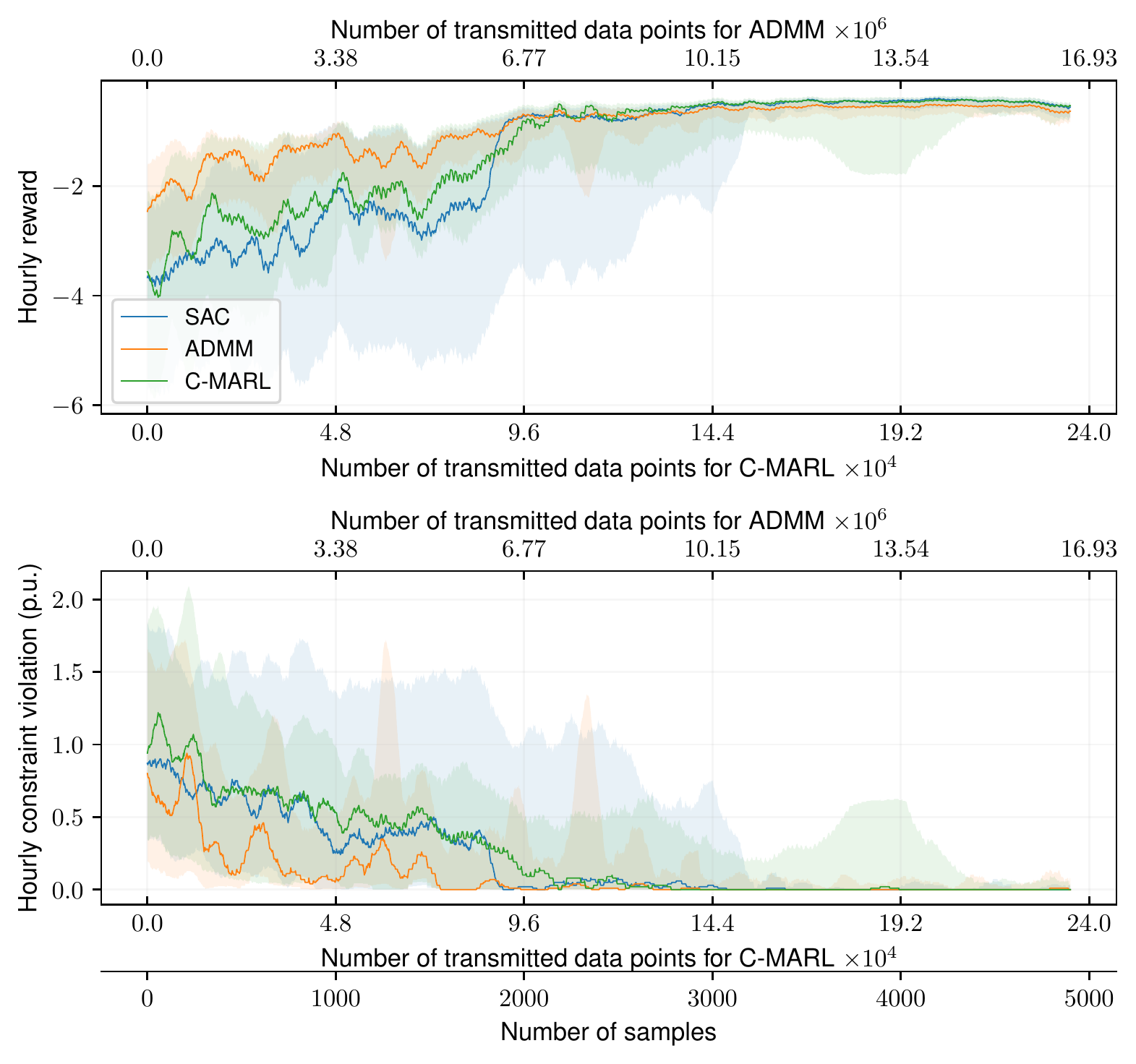}
	\caption{Hourly reward and voltage violation of 4-bus feeder}
	\label{fig_bus4}
\end{figure}
\begin{figure}[h]
	\centering
	\includegraphics[width=9cm]{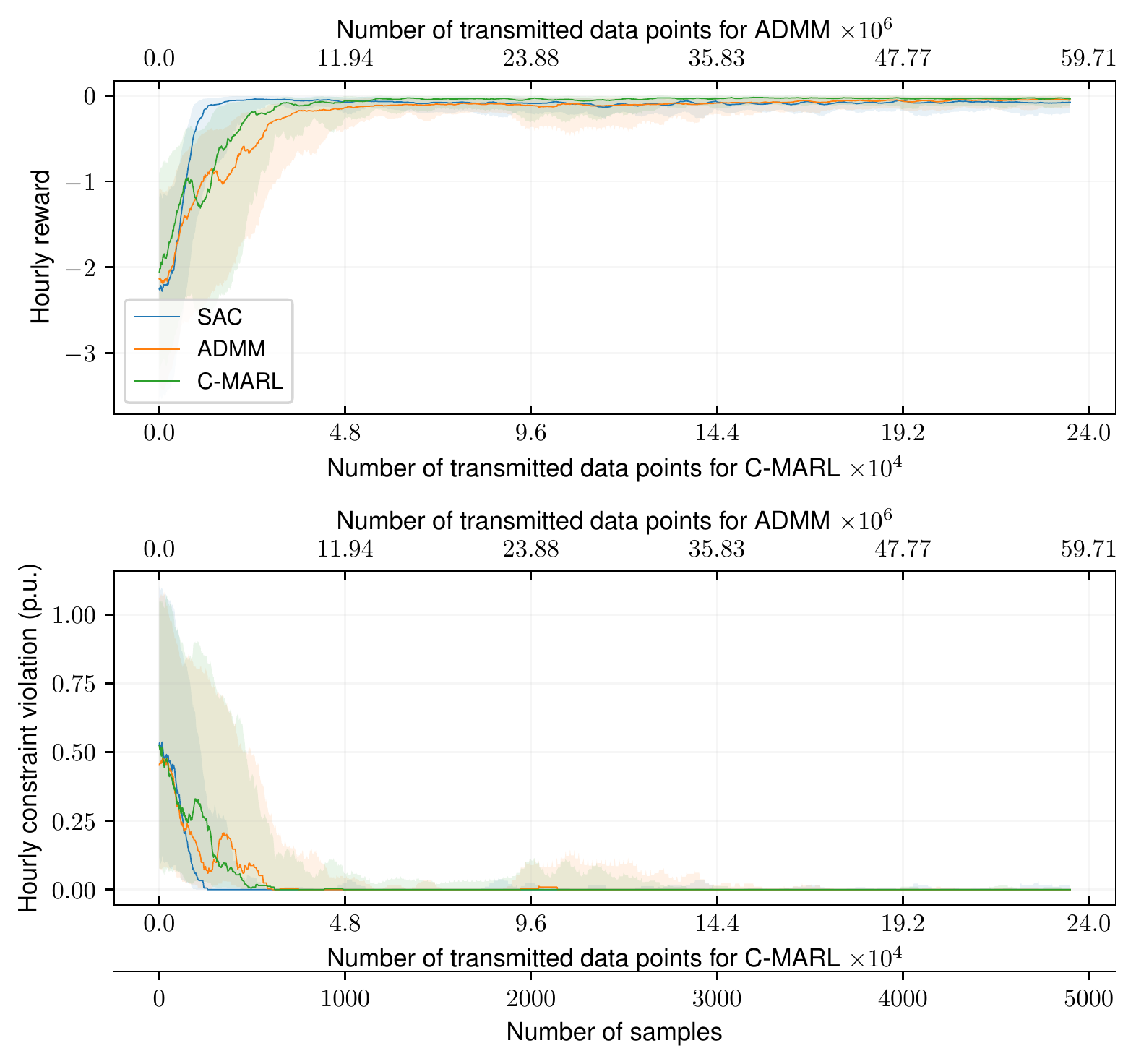}
	\caption{Hourly reward and voltage violation of 34-bus feeder}
	\label{fig_bus34}
\end{figure}
\begin{figure}[h]
	\centering
	\includegraphics[width=9cm]{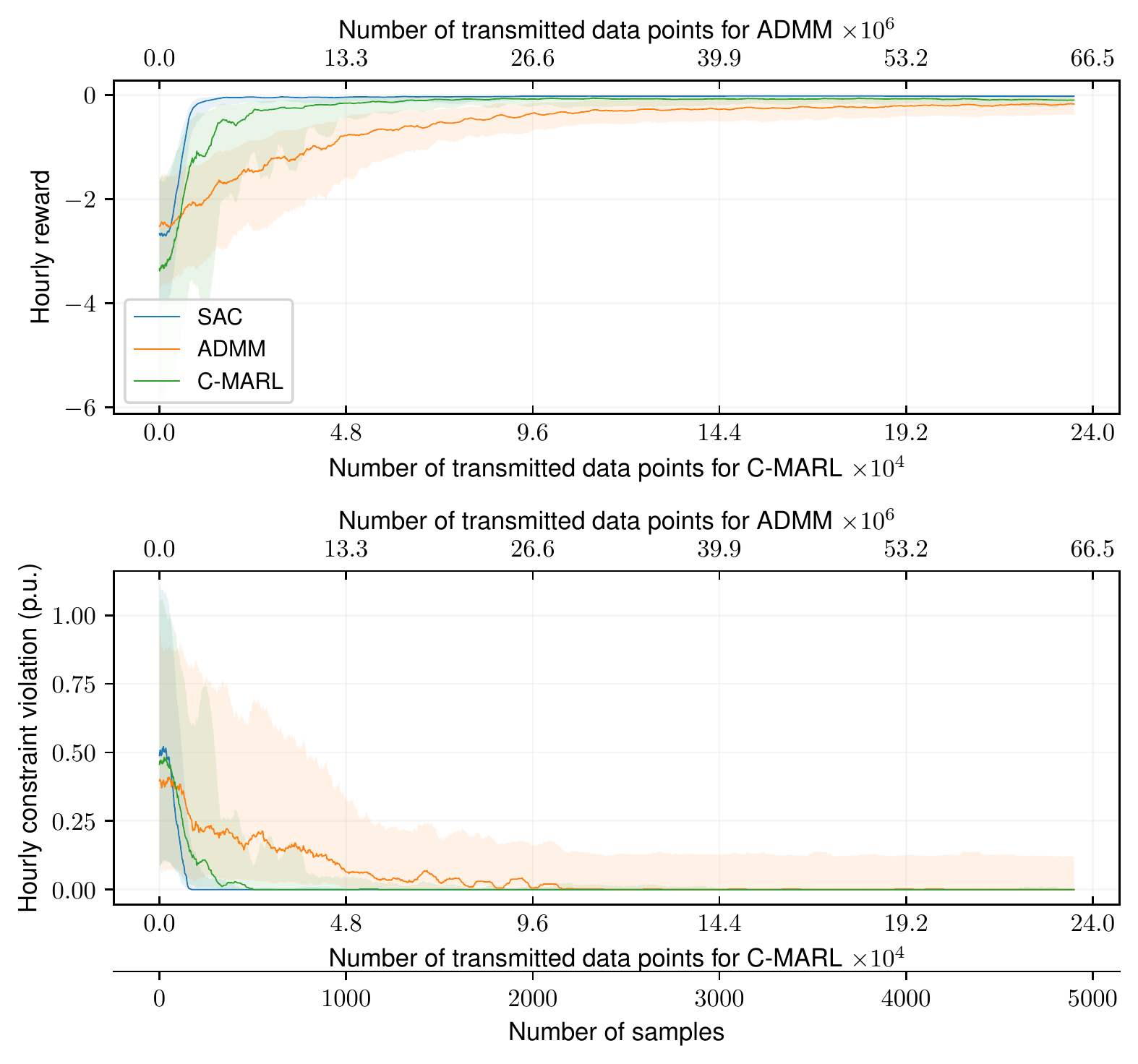}
	\caption{Hourly reward and voltage violation of 123-bus feeder}
	\label{fig_bus123}
\end{figure}

\subsection{Resiliency against Agent and Communication Link Failure}
One key advantage of distributed algorithms over the centralized ones is that when an individual agent or communication link fails, the rest of the system can continue to function. In this subsection, a few experiments are carried out to evaluate the proposed algorithm's resiliency against failures of individual components. Two types of component failures are considered:
\begin{enumerate}
    \item[E.1] An agent experiences an internal error so that the computation and control cannot be properly executed. However, it is still able to communicate with its neighbors. In this scenario, the agent freezes the tap position of its device and stops training, while other agents continue their controls and training.
    \item[E.2] A communication link is temporarily down. If the overall communication graph is still connected, then the agents function normally except for the altered communication graph connectivity. If the overall communication graph is disconnected, then the distribution network becomes partially observable. In this case, the agents will create a replacement state $\hat{S}_t=[\hat{\vb{p}}_t,\hat{\vb{q}}_t, \hat{A}_{t-1}, t]$ and take action based on $\hat{S}_t$. The nodal power $\hat{\vb{p}}_t,\hat{\vb{q}}_t$ are obtained from the historical average; the joint actions $\hat{A}_{t-1}$ are sampled from the agent's own policy network. Please note that each agent maintains a local copy of the joint policy network. The agent's experience involving replacement states won't be stored in the replay memory $\mc{D}$.
\end{enumerate}

We create two sets of experiments for the test feeders to demonstrate the consequences of the two types of failures E.1 and E.2. For all experiments, the occurrences of component failures are assumed to follow a Poisson process with rate $\lambda=\frac{1}{168}$. That is, the inter-event times are independent exponential random variables with scale parameter $\beta=\frac{1}{\lambda}=168$ (hr). The duration of each failure is assumed to follow the geometric distribution with success probability $0.2$. For the first experiment, all agents are assumed to have an equal chance of failure. The communication link failures in the second experiment are treated similarly. 

Simulation results for the two experiments are shown in Fig. \ref{fig_bothfailure}. Each experiment occupies one column. The blue curves represent the failure scenarios. The orange curves represent the corresponding counterfactual experiment, which has the identical simulation setup but without agent or communication failure. Fig. \ref{fig_bothfailure} shows that the proposed algorithm is resilient facing random agent or communication link failures. The algorithm performance degradation is negligible if the time to clear component failure is not too long. The impact of agent or communication failure on long-term algorithm performance is much smaller than that of short-term performance.
\begin{figure}[h]
	\centering
	\includegraphics[width=8.8cm]{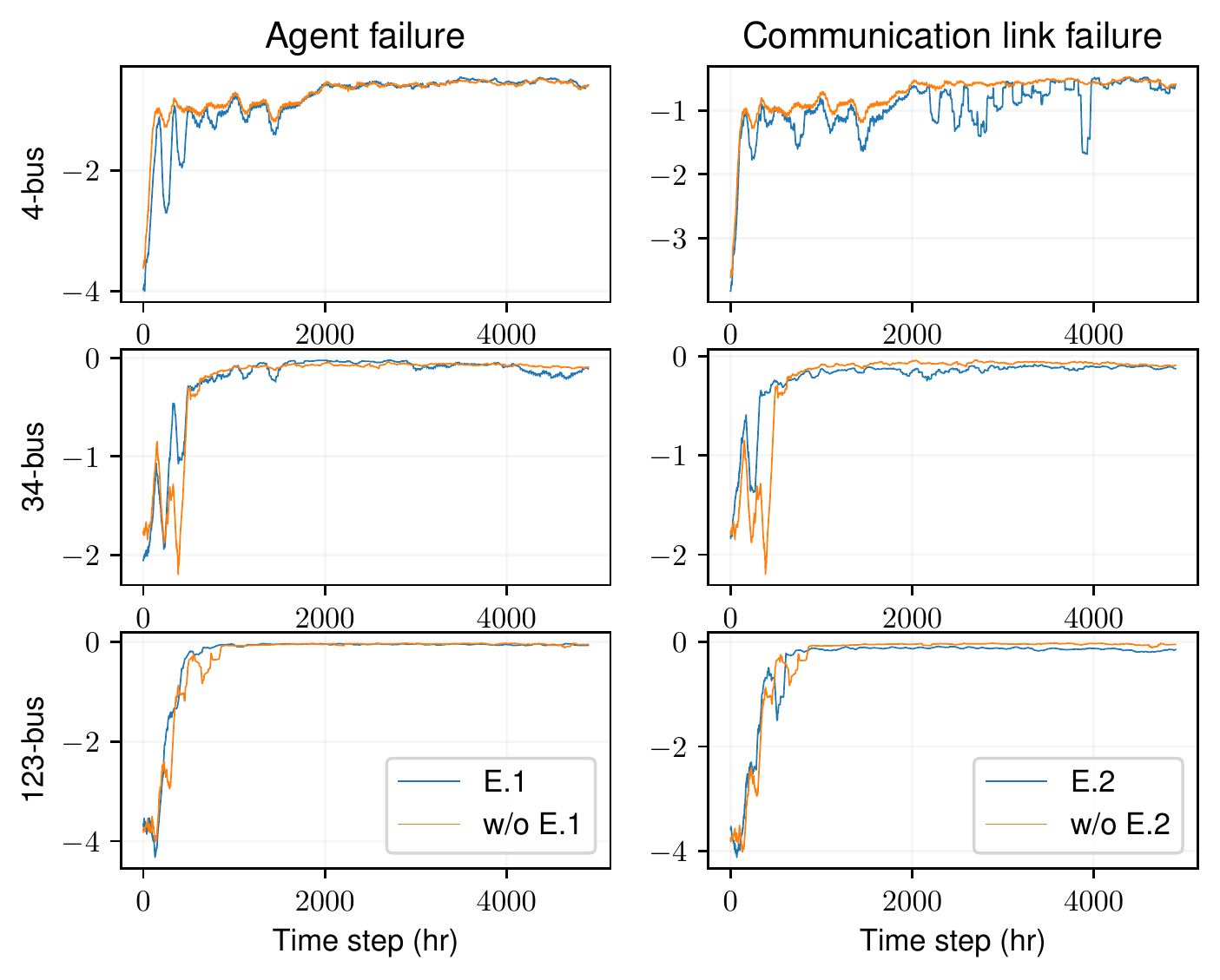}
	\caption{Hourly reward under agent/communication link failure}
	\label{fig_bothfailure}
\end{figure}

\section{Conclusion}
This paper proposes a multi-agent reinforcement learning algorithm to solve the Volt-VAR control problem in power distribution systems. We extend the centralized off-policy maximum entropy RL framework to a networked multi-agent MDP model. A randomization-based consensus algorithm is developed to solve the networked multi-agent MDP. Our proposed algorithm is decentralized and fully data-driven, which enables the control of voltage regulating devices without a central controller or knowledge of the distribution network topology and parameter information. Numerical study results of a comprehensive set of IEEE test feeders show that the proposed algorithm achieves a similar level of performance as the centralized RL benchmark. Our proposed algorithm is much more communication efficient than existing consensus strategy such as ADMM. Moreover, our proposed algorithm is resilient against communication link and agent failure as demonstrated by the simulation results.

\appendix[Proof of Propositions]

\begin{lemma}\label{lemma1}
Let $(X, d_X)$ and $(Y, d_Y)$ be two metric spaces and $(Y, \Sigma, \mu)$ be a measure space. Further, let $\Sigma$ be generated by the open sets in $(Y, d_Y)$; and $\mu$ has full support in the sense that $\mu(S)>0$ for all non-empty open sets $S$ in $\Sigma$. Let $f:X\times Y\mapsto \mathbb{R}^k,k\geq 1$ be a non-negative function that is continuous for every $x$. Then the two sets $C$ and $D$ are equal:
\begin{align*}
    C & = \{x | f(x,y)=0, \forall y\in Y\} \\
    D & = \{x | \int f(x,y)d\mu(y) =0\}
\end{align*}
\end{lemma}
Equalities and inequalities are understood to be element-wise.
\begin{proof}
It is clear that $C\subseteq D$ since the condition in $C$ implies that in $D$. To demonstrate that they are equal, let a point $x\notin C$, therefore $f_i(x,y)=c>0$ for some $y$ and some coordinate $i$ of $f$. Then by the continuity of $f_i(x, \cdot)$, for every $\epsilon>0$, there exists $\delta>0$, such that the condition $d(f_i(x,y^\prime),f_i(x,y))<\epsilon$ is satisfied for every $y^\prime\in B^\delta(y)\triangleq \{y^\prime| d_Y(y^\prime,y)<\delta\}$. Pick a small enough $\epsilon < c$, then we have
\begin{align*}
    \int f_i(x,y^\prime)d\mu(y^\prime) & \geq \int_{B^\delta(y)}f_i(x,y^\prime)d\mu(y^\prime) \\
    & \geq (c-\epsilon)\mu(B^\delta(y)) \\
    & > 0 
\end{align*}
The last inequality is due to the full support assumption of $\mu$. This shows that $x\notin D$. Therefore $C=D$.
\end{proof}

\begin{proof}[Proof of Proposition \ref{prop1}]
We identify the space of all neural network weights $W\subseteq \mathbb{R}^N$ with Euclidean metric as the $X$ space in Lemma \ref{lemma1}. We can identify the state-action space $\mc{S}\times \mc{A}$ as the $Y$ space by defining a metric $d$ on $\mc{S}\times \mc{A}$. The measure $\mu_{\mc{SA}}$ on the state-action space has the full support property stated in Lemma \ref{lemma1}. We identify the function $f$ as $\mb{h}( \bar{D}_{ii}\zeta_{\varphi_i}(s,a)-\bar{A}_i\bm{\zeta}_{\bm{\varphi}}(s,a))$ in the statement of Proposition \ref{prop1}. We further identify that the set $C$ and $D$ in Lemma \ref{lemma1} correspond to the constraint set in (\ref{mamdp_problem}) and the set expressed by (\ref{stochastic_representation}). By Lemma \ref{lemma1}, these two sets are equal.
\end{proof}

\begin{lemma}\label{lemma2}
Consider the situation in Lemma \ref{lemma1} except that $f$ may not be continuous, and that $\mu$ may not have full support. Then the following two sets are equal:
\begin{align*}
    C & = \{x| \mu( \{ y | f_i(x,y)>0  \; \text{for some}\; i\} ) = 0 \} \\
    D & = \{x | \int f(x,y)d\mu(y) =0\}
\end{align*}
\end{lemma}
\begin{proof}
Let $x\in C$, then $\int f(x,y)d\mu(y)=0$ since the set $f(x,y)>0$ has measure zero. Thus $x\in D$. Hence $C\subseteq D$. On the other hand, let $x\in D$, define two sets $N_1(x)=\{y|f_i(x,y)>0 \; \text{for some}\; i\}$ and $N_2(x)=\{y|f(x,y)=0\}$. Thus
\begin{align*}
    \int f(x,y)d\mu(y) & = 0\\
    & = \int_{N_1(x)} f(x,y)d\mu(y) + \int_{N_2(x)} f(x,y)d\mu(y) \\
    & = \int_{N_1(x)} f(x,y)d\mu(y)
\end{align*}
Thus $\mu(N_1(x)) = 0$ and therefore $x\in C$. Hence $D\subseteq C$. Combining the two directions shows that $C=D$. Therefore, every point $x$ in $D$ satisfies $f(x,y)=0$ for most of $y$ except for some $y$ with measure zero.
\end{proof}
\begin{proof}[Proof of Proposition \ref{prop2}]
Ditto as proof of Proposition \ref{prop1}. Since the measure $\mu_{\mc{SA}}$ can be arbitrary, it can be selected as full support which covers almost all $s,a$ pairs.
\end{proof}

%




\ifCLASSOPTIONcaptionsoff
  \newpage
\fi



\bibliographystyle{IEEEtran}
\bibliography{MAVVC}
\end{document}